%% file: ex_article.tex
\documentclass[onefignum,onetabnum]{siamonline171218}

\input{ex_shared}

\ifpdf
\hypersetup{
  pdftitle={Localized covariance estimation},
  pdfauthor={R. J. Webber and Matthias Morzfeld}
}
\fi

\myexternaldocument{ex_supplement}

\begin{document}

\maketitle

\begin{abstract}
    A major problem in numerical weather prediction (NWP) is the estimation of high-dimensional covariance matrices from a small number of samples.
    Maximum likelihood estimators cannot provide reliable estimates when the overall dimension is much larger than the number of samples.
    Fortunately, NWP practitioners have found ingenious ways to boost the accuracy of their covariance estimators by leveraging the assumption that the correlations decay with spatial distance.
    In this work, Bayesian statistics is used to provide a new justification and analysis of the practical NWP covariance estimators.
    The Bayesian framework involves manipulating distributions over symmetric positive definite matrices, and it leads to two main findings: 
    (i) the commonly used ``hybrid estimator'' for the covariance matrix has a naturally Bayesian interpretation;
    (ii) the very commonly used ``Schur product estimator'' is \emph{not} Bayesian, but it can be studied and understood within the Bayesian framework.
    As practical implications, the Bayesian framework shows how to reduce the amount of tuning required for covariance estimation, and it suggests that efficient covariance estimation should be rooted in understanding and penalizing \emph{conditional} correlations, rather than correlations.
\end{abstract}




\section{Introduction}
In this work, we provide new insights into the estimation of high-dimensional covariance matrices from a small number of samples.
Our work is motivated by numerical weather prediction (NWP), a setting in which covariance estimation arises naturally and at a vast scale.
The goal in NWP is to generate a set of weather forecasts based on global weather models and real-world observations of the Earth's atmosphere~\cite{BTB14}.
All the standard NWP techniques require estimating the covariance matrix for the near-term weather forecasts~\cite{L03,PZ15}, but the dimensionality of the covariance estimation problem in NWP is immense. 
A typical global weather model has billions of unknowns, which are updated by tens of millions of observations within less than six hours of computing time.
The large cost of running the weather model necessitates that the ensemble size (number of model integrations) is small compared to the number of unknown weather variables.
A typical ensemble size is $\leq 100$ and therefore six orders of magnitude smaller than the number of unknowns.

The sample covariance is not an accurate covariance estimator unless the ensemble size is larger than the number of unknowns \cite{BL08a,BL08}.
Therefore, covariance estimation at the extreme scale of NWP can only be accomplished by using additional information and tricks, commonly referred to as ``covariance localization''.
The basic idea is that covariances should decay with spatial distance.
On six hour time scales, the weather in La Jolla, California, is uncorrelated with the weather in Chicago, Illinois.
Localization, in its simplest form, means damping estimates of long-range of covariances because large magnitudes are caused by sampling error, not by the existence of long-range covariances~\cite{HM98,H01}.
Localization started off as an ad hoc procedure that perhaps grew out of desperation to make NWP work -- early NWP attempts using the ensemble Kalman filter \cite{E09Book}, for example, led to useless forecasts because of the large errors in the forecast covariances.
By now, however, localization has been widely accepted as a necessary ingredient within the NWP community~\cite{HWAS09}.

Not surprisingly, localization has been studied extensively.
In the NWP community, theoretical work has focused on
adaptive localization methods \cite{AL13} and theories for optimal localization \cite{F15, Menetrier15,MH22}, but implementing these techniques can be data-intensive.
In practice,
localization is often implemented using a Schur product estimator \cite{HM98,H01,OEtAl04}, a hybrid covariance estimator \cite{BS13,MA15,SHKB18} (also called a ``shrinkage'' estimator~\cite{NRS15,PSS21}),
or a combination of the two estimators \cite{MA15},
but there is little to no mathematical justification.
Meanwhile, the statistical community has introduced localized estimators with rigorous guarantees \cite{BL08a,FB07,BL08,CRZ16}.
However, these estimators are only guaranteed to work in the asymptotic limit as the ensemble size and the state dimension jointly grow to infinity.
It remains unclear which localized estimators work best with \emph{finite} ensemble size and finite state dimension.

To our surprise, localization is not typically understood from the Bayesian perspective,
even though localization is a naturally  Bayesian procedure: 
We estimate an unknown (the covariance matrix) based on limited data (the ensemble\slash forecast states) and enforce prior information about the problem structure (the spatial decay of covariances).
This paper is about describing the Bayesian perspective and explaining why this interpretation of covariance localization may be useful in practice.
Put simply, we ask and answer the following question: ``\emph{Are there any Bayesian prior distributions that lead to existing localization methods?}''
This question is equivalent to asking, ``\emph{Which existing localization methods are rooted in a Bayesian framework?}''

To answer this question, we consider two different Bayesian prior distributions.
First, we consider the inverse Wishart distribution \cite[Sec.~5.2]{press2005applied}, a classical distribution over covariance matrices that leads to the hybrid covariance estimator \cite{HM98,H01,OEtAl04,BS13,MA15,SHKB18} as a maximum a posteriori estimator.
Second, we consider a new ``quadratically constrained'' (QC) distribution, which forces the off-diagonal entries of the precision matrix (inverse covariance matrix) to be small.
We introduce the QC prior to study localization via Schur products \cite{BS13,MA15,SHKB18,NRS15,PSS21}, which is a common method for covariance estimation but is surprisingly \emph{not} Bayesian.
We show that the QC covariance estimator converges to the Schur product estimator as the localization strength parameter tends toward infinity.
In summary, our work provides a new Bayesian justification for two commonly used localization estimators in NWP.

The Bayesian framework is useful for several reasons.
First, the framework is designed for finite ensemble size and, thus, more practically applicable than statistical techniques that are largely asymptotic (e.g., \cite{BL08a,FB07,BL08,CRZ16}).

Second, the Bayesian framework suggests how to adjust the parameters in the hybrid estimator and Schur product estimator as the ensemble size changes.
Typically in an operational setting, this adjustment is done implicitly, since the localization is re-tuned as the ensemble grows larger or smaller.
Our Bayesian theory helps with reducing the amount of tuning required.

Last, the Bayesian theory may help to construct new localization estimators for the future. 
Every year, Earth models are becoming increasingly complex, e.g., coupled atmosphere, ocean and sea ice models, or seasonal to sub-seasonal forecast models, and data assimilation is also being extended to geomagnetic models \cite{FHetAl10,GMTK21}.
For all these models, traditional localization based on a single length scale parameter may no longer be appropriate.
The theoretical foundations laid here are not limited to a single length scale parameter and are more generally applicable.
As a central feature, our theory emphasizes building the localization scheme via the precision matrix that describes the conditional correlations between variables, not via the covariance matrix directly.

The rest of this paper is organized as follows.
\Cref{sec:review} reviews covariance estimation from the NWP perspective.
\Cref{sec:bayesian} analyzes covariance estimation from the Bayesian perspective.
\Cref{sec:numerical} numerically tests the predictions of the Bayesian theory.
\Cref{sec:theorems} proves mathematical theorems to support our Bayesian analysis.
\Cref{sec:conclusions} offers a summary and some conclusions.

Throughout this paper, we use bold lower case letters to refer to vectors and bold capital letters to refer to matrices.
The $ij$ entry of the matrix $\bm{A}$ is written $\bm{A}_{ij}$.
The determinant of a matrix $\bm{A}$ is written $|\bm{A}|$ and the Schur (element-wise) product of compatible matrices $\bm{A}$ and $\bm{B}$ is written $\bm{A} \circ \bm{B}$.
Last, the Frobenius norm of a matrix is written $\lVert \bm{A} \rVert_{\rm F} = \bigl(\sum_{i,j=1}^d |\bm{A}_{ij}|^2\bigr)^{1/2}$.


\section{A rapid review of NWP covariance estimation}
\label{sec:review}
We start by briefly describing how covariance estimation is accomplished in NWP.
To keep things simple, we assume that independent samples $\bm{x}_1, \ldots, \bm{x}_n \in \mathbb{R}^d$ are drawn from a mean-zero Gaussian distribution,
\begin{equation*}
    \bm{x}_i \sim \mathcal{N}(\bm{0}, \bm{\Sigma}).
\end{equation*}
We assume that ensemble size $n$ is much smaller than the dimension $d$.
Our goal is to estimate the $d \times d$ positive semidefinite covariance matrix $\bm{\Sigma}$ from $n \ll d$ samples.

A classical estimator for $\bm{\Sigma}$ is the maximum likelihood (ML) estimator
\begin{equation*}
    \hat{\bm{\Sigma}} = \argmax_{\bm{\Sigma}} 
     \prod_{i=1}^n p(\bm{x}_i | \bm{\Sigma}).
\end{equation*}
Under the mean zero assumption, the ML estimator is just the ``sample covariance'' or ``empirical covariance''
\begin{equation*}
    \hat{\bm{\Sigma}}^{\rm samp} =
    \frac{1}{n} \sum_{i = 1}^n \bm{x}_i \bm{x}_i^T.
\end{equation*}
The ML estimator is unbiased (mean $\bm{\Sigma}$), and as the number of data points $n$ approaches infinity, the ML estimator
converges to the true covariance $\bm{\Sigma}$ at the optimal rate  \cite[Ch.~8]{van2000asymptotic}
\begin{equation}
\label{eq:explicit_rate}
    \sqrt{n} (\hat{\bm{\Sigma}}^{\rm samp} - \bm{\Sigma})
    \stackrel{\mathcal{D}}{\rightarrow} \mathcal{N}(\bm{0}, \mathcal{I}(\bm{\Sigma})^{-1}),
\end{equation}
where $\mathcal{I}(\bm{\Sigma})^{-1}$ denotes the inverse Fisher information tensor.
This means that the ML estimator achieves the optimal $\mathcal{O}(1 \slash \sqrt{n})$ error scaling,
and the limiting distribution of $\sqrt{n} (\hat{\bm{\Sigma}}^{\rm samp} - \bm{\Sigma})$ is as tightly concentrated as possible.
However, since we work in a framework where $n \ll p$, the sample covariance is known to be inaccurate \cite{BL08a,BL08}.

\subsection{Schur product estimators}
Covariance localization is an approach for increasing the accuracy of the sample covariance when the ensemble size is small.
The basic idea is to damp long-range correlations based on the assumption that correlations decay with distance.
Localization can be implemented via a Schur product with a symmetric positive definite localization matrix $\bm{L}$:
\begin{equation}
\label{eq:schur}
    \hat{\bm{\Sigma}}^{\rm Schur} = \hat{\bm{\Sigma}}^{\rm samp} \circ \bm{L}.
\end{equation}
A simple example of a localization matrix is based on the Gaussian kernel and has elements
\begin{equation*}
    \bm{L}_{ij} 
    = \exp \bigl(-(d_{ij}\slash \ell)^2\bigr),
\end{equation*}
where $d_{ij}$ is the distance between grid points $i$ and $j$, and where $\ell > 0$ is a length scale parameter.
Alternately, one can replace the Gaussian kernel function with a Laplacian kernel function (exponential decay), or with a Gaspari-Cohn kernel function \cite{GC99}.
The latter is zero at large distances and therefore promotes sparsity in the  covariance matrix estimate.

Compared to the sample covariance, the Schur product estimator creates an element-wise bias of 
\begin{equation*}
    \mathbb{E} [\hat{\bm{\Sigma}}^{\rm Schur}_{ij}]
    - \bm{\Sigma}_{ij}
    = \bm{L}_{ij} \mathbb{E} [\hat{\bm{\Sigma}}^{\rm samp}_{ij}] - \bm{\Sigma}_{ij}
    = (\bm{L}_{ij} - 1) \bm{\Sigma}_{ij},
\end{equation*}
and changes the element-wise variance by a factor of
\begin{equation*}
    \frac{\textup{Var}[\hat{\bm{\Sigma}}^{\rm Schur}_{ij}]}{\textup{Var}[\hat{\bm{\Sigma}}^{\rm samp}_{ij}]}
    = \frac{\bm{L}_{ij}^2 \textup{Var}[\hat{\bm{\Sigma}}^{\rm samp}_{ij}]}{\textup{Var}[\hat{\bm{\Sigma}}^{\rm samp}_{ij}]}
    = \bm{L}_{ij}^2.
\end{equation*}
Since $\bm{L}_{i,j}$ is small at large spatial separations, the effect of localization is clear: it introduces a small bias while drastically reducing the variance of the estimator. 

\subsection{Hybrid estimators}
A second important and widely used covariance estimator is the ``hybrid'' estimator $\hat{\bm{\Sigma}}^{\rm hyb}$, which is defined as a convex combination between the sample covariance $\hat{\bm{\Sigma}}^{\rm samp}$ and a prior covariance estimate $\bm{\Sigma}^{\rm prior}$:
\begin{equation}
    \label{eq:hybrid}
    \hat{\bm{\Sigma}}^{\rm hyb}
    = \alpha \bm{\Sigma}^{\rm prior} + (1-\alpha) \hat{\bm{\Sigma}}^{\rm samp},
\end{equation}
for some $\alpha \in (0, 1)$.
Compared to the unbiased estimator $\hat{\bm{\Sigma}}^{\rm samp}$, the hybrid estimator creates a bias of size 
\begin{equation*}
    \mathbb{E} [\hat{\bm{\Sigma}}^{\rm hyb}] - \bm{\Sigma}
     = \alpha \bm{\Sigma}^{\rm prior} + (1-\alpha) \mathbb{E} [\hat{\bm{\Sigma}}^{\rm samp}] - \bm{\Sigma}
     = \alpha (\bm{\Sigma}^{\rm prior} - \bm{\Sigma})
\end{equation*}
and changes the variance by a factor of 
\begin{equation*}
    \frac{\textup{Var}[\hat{\bm{\Sigma}}^{\rm hyb}]}{\textup{Var}[\hat{\bm{\Sigma}}^{\rm samp}]}
    = \frac{(1-\alpha)^2 \textup{Var}[\hat{\bm{\Sigma}}^{\rm samp}]}{\textup{Var}[\hat{\bm{\Sigma}}^{\rm samp}]} = (1-\alpha)^2.
\end{equation*}
Thus, the hybrid estimator typically adds a small bias while slightly reducing the variance.

In NWP, the hybrid estimator is often presented in an equivalent form
\begin{equation}
    \label{eq:more_general}
    \hat{\bm{\Sigma}}^{\rm hyb,\,NWP} = w_1\bm{\Sigma}^{\rm clim} + w_2\hat{\bm{\Sigma}}^{\rm samp},
\end{equation}
where $\bm{\Sigma}^{\rm clim}$
is a climatological covariance matrix,
derived from a long model run that reveals covariance structure inherent to the physical process.
The NWP version of the hybrid estimator in \cref{eq:more_general} and the version we presented in \cref{eq:hybrid} are equivalent if we set 
\begin{equation*}
    \alpha = 1 - w_2, \qquad \bm{\Sigma}^{\rm prior} = \frac{w_1}{1 - w_2} \hat{\bm{\Sigma}}^{\rm clim}.
\end{equation*}
In \cref{sec:iw_prior}, we will show that NWP researchers are using a principled Bayesian approach when applying the hybrid estimator, but the Bayesian ideas are somewhat hidden within the notation. 
If we use the right symbols and notation, we can frame the practical NWP estimators within a rigorous Bayesian perspective.

\subsection{Tuning of covariance estimators}
\label{sec:tuning}
The accuracy of the Schur product and hybrid estimators depends on the various parameters that go into the construction.
For the hybrid estimator in \cref{eq:hybrid}, one needs to specify the prior covariance matrix $\bm{\Sigma}^{\rm prior}$ and the interpolation factor $\alpha$.
For the Schur product estimator in \cref{eq:schur}, one needs to specify the parameters that define the localization matrix $\bm{L}$.
If we use Gaussian or Laplacian kernels to define the localization matrix, this means that one needs to determine an appropriate length scale $\ell$ for localization.

The parameters that define the covariance estimator are usually determined via parameter tuning, or, using more modern language, a  ``training'' phase.
The idea is to simply try a few parameters and then determine which parameter combination gives the most useful results.
For example, one can run an ensemble data assimilation algorithm on a set of training observations and compute the forecast error that results from each choice of parameters.
One then selects the parameters that lead to the smallest forecast errors.

This tuning is expensive, computationally and otherwise.
In practice, localization and hybrid estimators are often combined \cite{MA15}, which means that a relatively large number of parameters needs to be tuned, which is even more costly.
Even worse, this entire tuning process must be repeated whenever the underlying model is modified, or if the ensemble size is increased because more computational power is available.
We will see in \cref{sec:bayesian} that the Bayesian perspective on covariance localization gives insights that can reduce the efforts that go into tuning covariance estimators.

\section{The Bayesian perspective on covariance estimation}
\label{sec:bayesian}

The main goal of Bayesian statistics is to combine prior information and data to estimate parameters in a model.
Bayesian statistics has three main components: the prior distribution, the likelihood, and the posterior distribution.
The Bayesian prior distribution encodes all information before any data are collected and the likelihood function infuses information from data into the posterior estimate.

Here, we apply Bayesian statistics to the problem of estimating a positive definite covariance matrix $\bm{\Sigma} \in \mathbb{R}^{d \times d}$ from a set of $n$ samples $\bm{x}_i$, $i=1,\dots,n$.
We assemble the $n$ samples into a $p\times n$ data matrix $\bm{X} = \begin{pmatrix} \bm{x}_1 & \cdots & \bm{x}_n \end{pmatrix}$,
and we express the posterior density function as
\begin{equation}
\label{eq:bayesian}
    \underbrace{p(\bm{\Sigma}| \bm{X})}_{\textup{posterior}}
    \propto \underbrace{p(\bm{\Sigma})}_{\textup{prior}}
    \underbrace{p(\bm{X} | \bm{\Sigma})}_{\textup{likelihood}}.
\end{equation}
The symbol $\propto$ indicates that the left- and right-hand sides are proportional over all choices of $\bm{\Sigma}$, but the proportionality constant is typically not needed for computing the covariance estimate.
In \cref{eq:bayesian}, the prior density function $p(\bm{\Sigma})$ is chosen to account for any structural knowledge of $\bm{\Sigma}$, e.g., the decay of correlations with spatial distance.
The likelihood function $p(\bm{X} | \bm{\Sigma})$ accounts for information from the data, which in our case are the $n$ samples assembled in the data matrix $\bm{X}$, and the likelihood function takes the form
\begin{equation}
\label{eq:invert}
    p(\bm{X} | \bm{\Sigma})
    = \Bigl|\frac{1}{2 \pi} \bm{\Sigma}^{-1}\Bigr|^{n \slash 2}
    \exp\Bigl(-\frac{1}{2} \sum_{i=1}^n \bm{x}_i^T \bm{\Sigma}^{-1} \bm{x}_i \Bigr).
\end{equation}
Last, the posterior density function $p(\bm{\Sigma}|\bm{X})$ gives a distribution of possible covariance matrices.
Using the posterior density, we can calculate the ``maximum a posterior'' (MAP) estimator
\begin{equation}
    \hat{\bm{\Sigma}}^{\rm MAP} = \argmax_{\bm{\Sigma}} p(\bm{\Sigma}| \bm{X}).
\end{equation}
The MAP estimator can be regarded as the single most likely value for the covariance under the posterior distribution.
Other estimators (e.g., mean, median) are equally valid, but are harder to compute or analyze.

With the uniform prior distribution $p(\bm{\Sigma}) = \text{Const.}$, the MAP estimator is the same as the ML estimator
\begin{equation}
    \hat{\bm{\Sigma}}^{\rm MAP} = \argmax_{\bm{\Sigma}} 
    p(\bm{X} | \bm{\Sigma})
    = \frac{1}{n} \sum_{i = 1}^n \bm{x}_i \bm{x}_i^T.
\end{equation}
However, more generally, the choice of prior distribution has a non-trivial effect on the Bayesian posterior
--- the whole purpose of imposing a prior is to fill in the gaps that the data leave open.
The rest of this paper is about choices of non-uniform priors $p(\bm{\Sigma})$ that promote structure in the covariance estimate $\hat{\bm{\Sigma}}^{\rm MAP}$,
providing justification for existing but largely empirical covariance estimators in NWP.

\subsection{The inverse Wishart prior and hybrid estimators} \label{sec:iw_prior}

The inverse Wishart distribution \cite[Sec.~5.2]{press2005applied} is a classical distribution defined over symmetric positive definite matrices $\bm{\Sigma} \in \mathbb{R}^{d \times d}$ by the density
\begin{equation}
\label{eq:wishart}
    p(\bm{\Sigma}) \propto \bigl| \bm{\Sigma}^{-1} \bigr|^{m \slash 2}
    \exp\Bigl(-\frac{m}{2} \textup{tr} \bigl(\bm{\Sigma}^{\rm prior} \bm{\Sigma}^{-1} \bigr)\Bigr).
\end{equation}
There are two parameters in the inverse Wishart distribution:
$\bm{\Sigma}^{\rm prior}$ is the mode (most likely value) of the distribution, and $m$ is the ``sample size'' parameter that controls the width of the distribution around the mode:
a large $m$ leads to tightly concentrated distribution.

Given an inverse Wishart prior and a data matrix $\bm{X} = \begin{pmatrix} \bm{x}_1 & \cdots & \bm{x}_n \end{pmatrix}$, the Bayesian posterior is also an inverse Wishart distribution,
because the inverse Wishart distribution is the ``conjugate prior'' \cite{diaconis1979conjugate} to the mean-zero multivariate Gaussian likelihood:
\begin{align*}
    p(\bm{\Sigma} | \bm{X})
    &\propto
    \bigl| \bm{\Sigma}^{-1} \bigr|^{(m+n) \slash 2}
    \exp\Bigl(-\frac{m}{2} \textup{tr} \bigl(\bm{\Sigma}^{\rm prior} \bm{\Sigma}^{-1} \bigr)
    -\frac{1}{2} \sum_{i=1}^n \bm{x}_i^T \bm{\Sigma}^{-1} \bm{x}_i \Bigr) \\
    \label{eq:conjugate}
    & \propto \bigl| \bm{\Sigma}^{-1} \bigr|^{(m + n) \slash 2}
    \exp\Bigl(-\frac{m + n}{2} \textup{tr} \bigl(\hat{\bm{\Sigma}}^{\rm IW} \bm{\Sigma}^{-1} \bigr)\Bigr),
\end{align*}
where
\begin{equation}
\label{eq:IWMAP}
    \hat{\bm{\Sigma}}^{\rm IW} = \frac{m}{m + n} \bm{\Sigma}^{\rm prior} + \frac{n}{m + n} \hat{\bm{\Sigma}}^{\rm samp}.
\end{equation}
In the posterior distribution, the two inverse Wishart parameters are updated in response to the data:
the sample size parameter increases from $m$ to $m + n$,
and the mode changes from $\bm{\Sigma}^{\rm prior}$ to $\hat{\bm{\Sigma}}^{\rm IW}$ \cref{eq:IWMAP}.

It is now clear that the inverse Wishart prior leads to a covariance estimator \cref{eq:IWMAP} that is identical to the hybrid estimator \cref{eq:hybrid} with the parameter choice
\begin{equation}
\label{eq:ShrinkageAlpha}
\alpha = \frac{m}{m + n}.
\end{equation}
In other words, the hybrid estimator is the same estimator that would result from selecting an inverse Wishart prior and systematically applying a Bayesian analysis.
This perspective provides a Bayesian justification for the hybrid estimator, assuming that the parameters $m$ and $\bm{\Sigma}^{\rm prior}$ represent reasonable prior knowledge about the covariance structure.

A major benefit of Bayesian statistics is that it leads to a covariance estimator \cref{eq:IWMAP} valid for any sample size $n$, whereas standard NWP covariance estimators require tuning parameters whenever the sample size changes (\cref{sec:tuning}).
When $n$ is large, the Bayesian formula tells us to adjust our estimator according to
\begin{equation}
\label{eq:n_scaling}
    \hat{\bm{\Sigma}}^{\rm IW} 
    = \hat{\bm{\Sigma}}^{\rm samp} + \mathcal{O}(n^{-1}),
\end{equation}
and this scaling with $n$ ensures that $\hat{\bm{\Sigma}}^{\rm IW}$ converges to the true covariance at the optimal asymptotic rate \cref{eq:explicit_rate} as $n \rightarrow \infty$.
In NWP applications,
we anticipate consistent accuracy in covariance estimation when the localized covariance estimate is adjusted according to \cref{eq:IWMAP,eq:n_scaling}.
We will revisit this idea in the numerical examples in \cref{sec:numerical}.

\subsection{The QC prior and Schur product estimators} 

We now introduce a new ``quadratically contrained'' (QC) distribution to study localization via Schur products.
Surprisingly, localization via Schur products \emph{cannot} result directly from a Bayesian prior (\cref{prop:not_bayesian}).
However, the QC prior allows us to study Schur product localization in a rigorous and meaningful way, in the asymptotic limit of increasing penalization strength.

The QC distribution is defined over symmetric positive definite matrices by the density function
\begin{equation*}
    p(\bm{\Sigma}) \propto
    \exp\Bigl(- \frac{1}{4} \textup{tr}\bigl(\bm{\Sigma}^{-1} \bigr(\bm{\Theta} \circ \bm{\Sigma}^{-1}\bigr)\bigr) \Bigr).
\end{equation*}
The only parameter in the QC distribution is a symmetric nonnegative-valued matrix $\bm{\Theta}$.
The QC prior can be ``improper'' \cite{dawid1973marginalization}, i.e.,
the density can integrate to infinity for some $\bm{\Theta}$. 
However, the corresponding Bayesian posterior distribution
\begin{equation}
\label{eq:posterior}
    p(\bm{\Sigma}|\bm{X})
    \propto \bigl| \bm{\Sigma}^{-1} \bigr|^{n \slash 2}
    \exp\Bigl(-\frac{n}{2} \textup{tr} \bigl(\hat{\bm{\Sigma}}^{\rm samp} \bm{\Sigma}^{-1} \bigr) - \frac{1}{4} \textup{tr}\bigl(\bm{\Sigma}^{-1} \bigl( \bm{\Theta} \circ \bm{\Sigma}^{-1}\bigr)\bigr) \Bigr).
\end{equation}
is well-defined for every $\bm{\Theta}$.
We further show in \cref{prop:maximum} that this posterior distribution has a unique positive definite global maximizer given a large localization strength, which justifies the use of the MAP as a covariance estimator.

\subsubsection{Motivation for the QC prior}

The QC distribution can be derived as the maximum entropy or most ``random'' \cite{jaynes1957information} distribution that constrains the square entries of the precision matrix $\bm{\Sigma}^{-1}$ to be small (\cref{prop:max_entropy}).
More specifically, with entropy defined as  the amount of ``randomness'' in a density $p$ via
\begin{equation*}
    H[p] := - \int p(\bm{\Sigma}) \log p(\bm{\Sigma}) \mathop{d\bm{\Sigma}},
\end{equation*}
the QC density solves the maximization problem
\begin{equation*}
    \max_p 
    \biggl\{
    H[p]
    - \frac{1}{4} \sum_{i,j = 1}^d \bm{\Theta}_{ij}\, \int p(\bm{\Sigma}) |\bm{\Sigma}^{-1}_{ij}|^2 d\bm{\Sigma} \biggr\}.
\end{equation*}
Here, $\bm{\Theta}$ is the parameter that penalizes off-diagonal elements in $\bm{\Sigma}^{-1}$.

At first, it is perhaps strange that we define the QC prior to target off-diagonal elements in the \emph{precision} matrix $\bm{\Sigma}^{-1}$, while we aim to explain the Schur product estimator that constrains elements in the \emph{covariance} matrix $\bm{\Sigma}$.
However, there is a systematic Bayesian explanation for why targeting the precision matrix is the right approach, based on the conditional correlation structure.

The conditional correlation between two variables $\bm{x}_i$ and $\bm{x}_j$ measures the degree of association with the effects of all other components of $\bm{x}$ removed.
In many NWP applications, we expect that conditional correlations, even more so than correlations, should be confined to small neighborhoods.
For example, we expect the weather in La Jolla, California is conditionally uncorrelated with the weather in Chicago, Illinois, after accounting for the weather in all the in-between locations.
The fast decay of conditional correlations has been observed in many geophysical applications and has been described as the ``screening effect''  \cite{stein2002screening}.

In a Gaussian model, the conditional correlations are described explicitly by
\begin{equation*}
    \textup{corr}\bigl[\,\bm{x}_i, \bm{x}_j \,|\, (\bm{x}_k)_{k \notin \{i, j\}}\,\bigr] = -\frac{\bm{\Sigma}^{-1}_{ij}}{(\bm{\Sigma}^{-1}_{ii} \bm{\Sigma}^{-1}_{jj})^{1/2}}.
\end{equation*}
The magnitude of the conditional correlations is thus proportional to the magnitude of the $\bm{\Sigma}^{-1}$ elements.
The QC prior can be interpreted as enforcing prior knowledge of the screening effect, by targeting the off-diagonal entries of $\bm{\Sigma}^{-1}$.

As an example of the screening effect, we consider a Gaussian process with covariances defined by the Laplacian kernel 
\begin{equation*}
    k(x, y) = \exp\Bigl(-\frac{|x - y|}{\ell}\Bigr),
\end{equation*}
on a 1D spatial domain (not periodic).
When the data is generated from a uniform grid with mesh size $\Delta$, the corresponding covariance matrix is
\begin{equation*}
    \bm{\Sigma}
    = \begin{pmatrix}
    1 & e^{-\Delta} & \cdots & e^{-(d-2)\Delta} & e^{-(d-1)\Delta}  \\
    e^{-\Delta} & 1 & \cdots & e^{-(d-3)\Delta} & e^{-(d-2)\Delta} \\
    \vdots & \vdots & & \vdots & \vdots \\
    e^{-(d-2)\Delta} & e^{-(d-3)\Delta} & \cdots & 1 & e^{-\Delta} \\
    e^{-(d-1)\Delta} & e^{-(d-2)\Delta} & \cdots & e^{-\Delta} & 1
    \end{pmatrix},
\end{equation*}
and the precision matrix is
\begin{equation*}
    \bm{\Sigma}^{-1}
    = \frac{2}{e^{\Delta} - e^{-\Delta}} \begin{pmatrix}
    e^{\Delta} & -1 \\
    -1 & e^{\Delta} + e^{-\Delta} & \ddots \\
    & \ddots & \ddots & \ddots \\
    && \ddots & e^{\Delta} + e^{-\Delta} & -1 \\
    && & -1 & e^{\Delta}
    \end{pmatrix}.
\end{equation*}
The precision matrix $\bm{\Sigma}^{-1}$ is tridiagonal and, hence, has a faster off-diagonal decay than the covariance matrix (which has exponential decay).
This example thus supports the strategy of constraining off-diagonal entries in $\bm{\Sigma}^{-1}$, rather than in $\bm{\Sigma}$.


\subsubsection{QC covariance estimator}

Next, we study the MAP estimator corresponding to the QC prior.
We do so by maximizing the logarithm of the posterior distribution \cref{eq:posterior}
\begin{equation*}
    \ell(\bm{\Sigma}) 
    = \frac{n}{2} \log \bigl| \bm{\Sigma}^{-1} \bigr| - \frac{n}{2} \textup{tr}\bigl(\hat{\bm{\Sigma}}^{\rm samp} \bm{\Sigma}^{-1}\bigr) - \frac{1}{4} \textup{tr}\bigl(\bm{\Sigma}^{-1} \bigl( \bm{\Theta} \circ \bm{\Sigma}^{-1} \bigr) \bigr).
\end{equation*}
To find the unique global maximizer of $\ell$ we set its gradient
\begin{equation*}
    \nabla \ell = \frac{n}{2} \bm{\Sigma}^{-1}
    \Bigl[ -\bm{\Sigma} + \hat{\bm{\Sigma}}^{\rm samp} + \frac{1}{n} \bm{\Sigma}^{-1} \circ \bm{\Theta} \Bigr] \bm{\Sigma}^{-1}
\end{equation*}
equal to zero, and we obtain an implicit equation for the QC estimator
\begin{equation}
\label{eq:nonlinear}
    \bm{\Sigma}^{\rm QC} = \hat{\bm{\Sigma}}^{\rm samp}
    + \frac{1}{n} (\bm{\Sigma}^{\rm QC})^{-1} \circ \bm{\Theta}.
\end{equation}
In high dimensions, solving \cref{eq:nonlinear} is a challenge.
Nonetheless, we can extract useful asymptotic information from \cref{eq:nonlinear} and make the connection to Schur product estimators.

First, we note that
\cref{eq:nonlinear} implies the QC estimator is the same as the sample covariance
\begin{equation*}
    \hat{\bm{\Sigma}}_{ij}^{\rm QC} = \hat{\bm{\Sigma}}_{ij}^{\rm samp}.
\end{equation*}
for any $(i,j)$ entries such that $\bm{\Theta}_{ij} = 0$.
In other words, the QC estimator trusts the sample covariance completely if we do not penalize the  conditional correlation between $\bm{x}_i$ and $\bm{x}_j$ via $\bm{\Theta}_{i,j} > 0$.
In NWP, it is unusual to penalize variances, so we assume for the rest of this section that $\bm{\Theta}_{ij} > 0$ if and only if $i \neq j$, which implies that 
\begin{equation*}
    \hat{\bm{\Sigma}}_{ii}^{\rm QC} = \hat{\bm{\Sigma}}_{ii}^{\rm samp}, \qquad 1 \leq i \leq p.
\end{equation*}

We now consider the asymptotic behavior of the QC estimator when we set $\bm{\Theta} = s \bm{\Theta}^{\rm ref}$ and raise the penalization strength parameter $s \rightarrow \infty$.
In this limit, we may write
\begin{equation}
\label{eq:ApproxMAP}
    \hat{\bm{\Sigma}}^{\rm QC}
    = \bm{D} + s^{-1} \bm{\Delta},
\end{equation}
where $\bm{D}$ is a diagonal matrix with elements $\bm{D}_{ii} = \hat{\bm{\Sigma}}_{ii}^{\rm QC} = \hat{\bm{\Sigma}}_{ii}^{\rm samp}$ 
and $s^{-1} \bm{\Delta}$ is the matrix containing all off-diagonal elements of $\hat{\bm{\Sigma}}^{\rm QC}$.
As $s \rightarrow \infty$, the inverse of the QC estimator is given by the Taylor series expansion
\begin{equation}
\label{eq:ApproxMAPInv}
    (\hat{\bm{\Sigma}}^{\rm QC})^{-1} = \bm{D}^{-1} - s^{-1} \bm{D}^{-1} \bm{\Delta} \bm{D}^{-1} + \mathcal{O}(s^{-2}).
\end{equation}
Substituting~\cref{eq:ApproxMAP} and~\cref{eq:ApproxMAPInv} into~\cref{eq:nonlinear},
and solving for $s^{-1} \bm{\Delta}$, we find 
\begin{equation*}
    \hat{\bm{\Sigma}}_{ij}^{\rm QC} = \hat{\bm{\Sigma}}^{\rm samp}_{ij} \biggl(1 + \frac{\bm{\Theta}_{ij}}
    {n \hat{\bm{\Sigma}}^{\rm samp}_{ii} \hat{\bm{\Sigma}}^{\rm samp}_{jj}}\biggr)^{-1} + \mathcal{O}(s^{-2}).
\end{equation*}
Therefore, the QC covariance estimator is asymptotically a Schur product estimator that damps the off-diagonal entries of the covariance matrix.

The Bayesian framework thus begins to provide a theoretical foundation for the Schur product estimator used in NWP.
The QC estimator converges to a Schur product estimator with localization matrix
\begin{equation}
\label{eq:schur_formula}
    \bm{L}_{ij} = \biggl(1 + \frac{\bm{\Theta}_{ij}}
    {n \hat{\bm{\Sigma}}^{\rm samp}_{ii} \hat{\bm{\Sigma}}^{\rm samp}_{jj}}\biggr)^{-1}.
\end{equation}
This approximation becomes increasingly accurate in the limit of a large penalization strength.

\subsubsection{Adjusting the length scale}

When we use a localization matrix $\bm{L}$ with a single length scale parameter, the Bayesian perspective suggests how best to adjust the length scale with the ensemble size $n$.
As $n \rightarrow \infty$, the Schur product formula \cref{eq:schur_formula} satisfies
\begin{equation}
\label{eq:ensures}
    \bm{L}_{ij} = 1 + \mathcal{O}(n^{-1}).
\end{equation}
For example, if $\bm{L}_{ij} = \exp(-d_{ij} \slash  \ell)$,
we need to adjust the length scale as
\begin{equation}
\label{eq:SchurScalingL}
    \exp(- d_{ij}\slash \ell) = 1 + \mathcal{O}(n^{-1}) \implies
    \frac{1}{\ell} = \mathcal{O}(n^{-1}),
\end{equation} 
to ensure that \cref{eq:ensures} is satisfied.
Thus, we need to increase the length scale at a rate $\ell \sim n$ or faster.
If $\bm{L}_{ij} = \exp(-d_{ij}^2 \slash  \ell^2)$,
we need to adjust the length scale as
\begin{equation}
\label{eq:SchurScalingG}
    \exp\bigl(-(d_{ij}\slash \ell)^2\bigr) 
    = 1 + \mathcal{O}(n^{-1}) \implies
    \frac{1}{\ell^2} = \mathcal{O}(n^{-1}).
\end{equation}
This means we need to increase the length scale at a rate $\ell \sim n^{1 \slash 2}$ or faster.

The scaling of the length scale with the ensemble size is practically important in NWP.
Typically when the ensemble size changes, the localization length scale(s) are re-tuned, which is costly, both computationally and otherwise (\cref{sec:tuning}).
The Bayesian perspective naturally provides a scaling of the localization length scale with ensemble size, thus reducing the amount of tuning necessary when the data assimilation system undergoes an upgrade.

\section{Numerical illustration} \label{sec:numerical}

In this section, we numerically test the scalings for the hybrid estimator and Schur product estimator that are predicted by our Bayesian theory.
We perform numerical tests with a variety of covariance matrices and a variety of ensemble sizes to check that the scalings are useful in practice
and to reiterate that the theory holds at finite ensemble size, not only asymptotically (for $n,p\to\infty$).

\subsection{Covariance matrices}
We estimate covariance matrices for five different Gaussian models, each with mean zero and a (spatial) dimension of $d = 200$.
Following \cite{MH22}, the five covariance matrices are defined as follows.
\begin{enumerate}
    \item
    \emph{Single-scale, Laplacian kernel}. The first covariance matrix has elements
    \begin{equation*}
        \bm{\Sigma}^{\ell}_{ij} 
        = \exp( -d_{ij}/\ell),
    \end{equation*}
    where $\ell$ is the length scale, which we set to $\ell = 5$, and $d_{ij}$ is the distance between grid points $i$ and $j$, reflecting the periodic domain.
    \item
    \emph{Single-scale, Gaussian kernel}. The second covariance matrix is similar to the first, but the decay of covariance is faster. Namely, the matrix has elements
    \begin{equation*}
        \bm{\Sigma}^{\ell}_{ij} 
        = \exp(-(d_{ij}/\ell)^2)
    \end{equation*}
    where $\ell$ and $d_{ij}$ are the length scale and the distance between grid points on the periodic domain.
    Here too, we set $\ell = 5$.
    \item
    \emph{Multiscale covariance}. We define a multiscale covariance  matrix as the sum of two single-scale covariances
    \begin{equation*}
        \bm{\Sigma}^{\rm ms} = \tfrac{1}{2} \bigl( \bm{\Sigma}^{\ell_1} + \bm{\Sigma}^{\ell_2} \bigr),
    \end{equation*}
    and we set $\ell_1=2$ and $\ell_2=20$.
    We use single-scale covariance matrices defined by a Gaussian kernel, but similar results can be obtained with a Laplacian kernel.
    \item \emph{Nonstationary covariance}.
    The next covariance matrix is ``nonstationary'' \cite{PS03} and its elements are 
    \begin{equation*}
        \bm{\Sigma}^{\text{ns}}_{{i,j}} = \frac{(4 \ell_i \ell_j)^{1/4}}{(\ell_i + \ell_j)^{1/2}} \exp\Bigl( - \frac{2 |i-j|^2}{\ell_i + \ell_j}\Bigr),   
    \end{equation*}
    where $\sqrt{\ell_i}$ can be thought of as a local length scale.
    Here, we consider the case where $\ell_i$ increases linearly from $2.1$ to $22$ over the domain (not periodic).
    \item \emph{Pressure-wind covariance}.
    The final covariance matrix models two spatially extended fields, pressure and wind, that co-vary with each other according to 
    \begin{equation}
	    \label{eq:presswind}
	    w = \frac{\text{d}u}{\text{d}x},
    \end{equation}
    where $u$ is pressure and $w$ is wind.
    The pressure has a single-scale covariance $\bm{\Sigma}^\ell$ ($\ell=5$, Gaussian kernel). \Cref{eq:presswind} then implies that the covariance matrix of both variables (pressure first and then wind) is
    \begin{equation*}
	    \bm{\Sigma}^\text{pw} = \begin{bmatrix}
		\bm{\Sigma}^\ell & \mathbf{D}\bm{\Sigma}^\ell \\ \\\bm{\Sigma}^\ell \mathbf{D}^T & \mathbf{D}\bm{\Sigma}^\ell \mathbf{D}^T
	    \end{bmatrix}
    \end{equation*}
    where $\bm{D} \in \mathbb{R}^{200\times 200}$ is a periodic, centered, second-order discretization of the first derivative operator.  
    Note that pressure and wind both have dimension 200 so that the overall dimension for this problem is 400.   
\end{enumerate}

\Cref{fig:CovFig} shows four of the five covariance matrices used in our numerical experiments.
The figure does not include a plot of the single-scale covariance with a Laplacian kernel because it looks similar to the single-scale covariance matrix with a Gaussian kernel.

\begin{figure}[t]
	\centering
	\includegraphics[width=.8\textwidth]{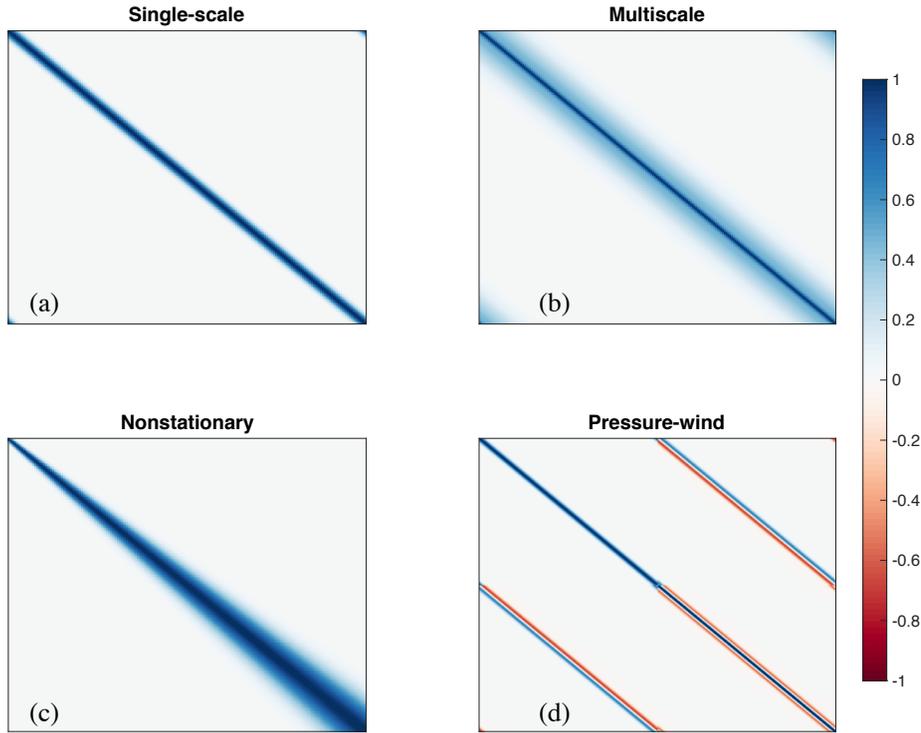}
	\caption{
	Four of the five covariance matrices used in numerical experiments.
	}
	\label{fig:CovFig}
\end{figure}

\subsection{Benefits of localized covariance estimators}

To illustrate the benefit of localization,
we consider estimating a single-scale covariance matrix (Gaussian kernel) of dimension $200\times 200$ from $n=30$ samples.
\Cref{fig:LocIllu} shows the true covariance matrix, along with the sample covariance and the hybrid and Schur product estimates.
The parameter tuning for the hybrid and Schur product estimators is described in \cref{sec:hybrid_experiments,sec:schur_experiments}.

\begin{figure}[t]
	\centering
	\includegraphics[width=.8\textwidth]{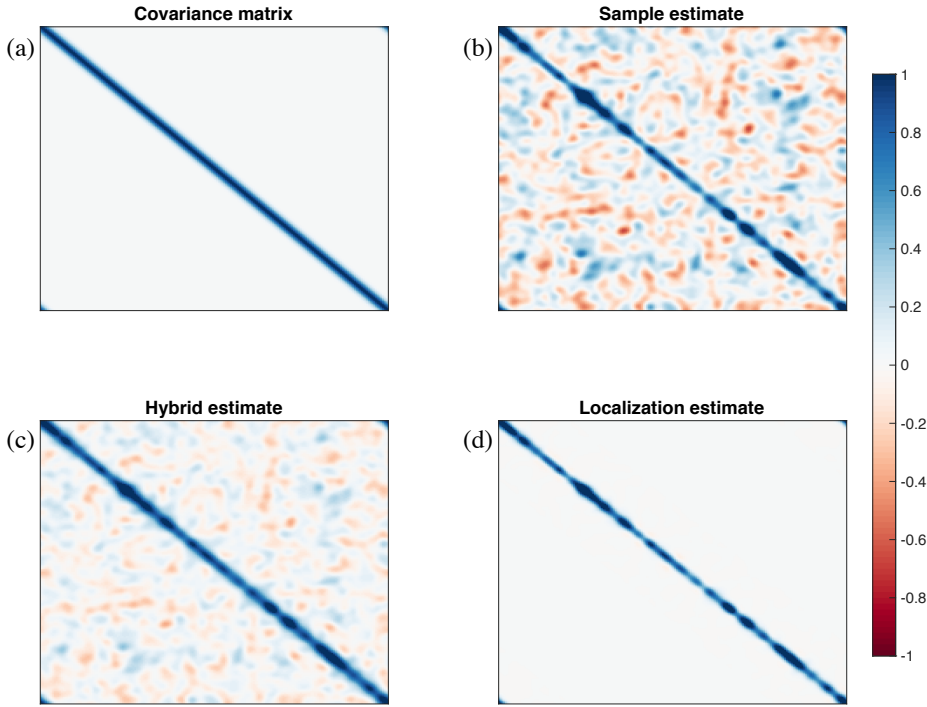}
	\caption{The true covariance matrix that we want to estimate from $n=30$ samples is shown in the top left.
	The sample covariance, shown in the top right, contains large errors because the ensemble size is small compared to the dimension $n=30\ll d=200$.
	The hybrid (bottom left) and Schur product (bottom right) localized covariance estimates are more accurate.}
	\label{fig:LocIllu}
\end{figure}

The figure makes it intuitively clear that hybrid and Schur product estimates are more accurate than the sample covariance.
The reason is that localized estimates take into account additional information about the covariance structure.
In \cref{sec:bayesian}, we explained how this additional information can be understood as a Bayesian prior distribution that is strongly influencing the covariance estimates.

\subsection{Experiments with hybrid estimators} \label{sec:hybrid_experiments}
We now turn our attention to the hybrid estimator
\begin{equation*}
    \hat{\bm{\Sigma}}^{\rm hyb}
    = \alpha \bm{\Sigma}^{\rm prior} + (1-\alpha) \hat{\bm{\Sigma}}^{\rm samp},
\end{equation*}
and evaluate whether or not it is appropriate to scale the interpolation factor $\alpha$ as
$\alpha \sim n^{-1}$,
which is suggested by the Bayesian theory \cref{eq:ShrinkageAlpha}.

We use the following procedure.
For each covariance matrix on our list, we apply interpolation factors $\alpha$ ranging from $0$ to $1$, in steps of $0.05$.
Then, we compute the error
\begin{equation*}
    \text{Error}(\alpha) := \frac{\lVert \hat{\bm{\Sigma}}^{{\rm hyb,}\,\alpha} - \bm{\Sigma}\rVert_{\rm F}}{\lVert \bm{\Sigma} \rVert_{\rm F}},
\end{equation*}
where $\hat{\bm{\Sigma}}^{{\rm hyb,}\,\alpha}$ is the hybrid estimate of the covariance matrix $\bm{\Sigma}$ using the interpolation factor $\alpha$.
We repeat this procedure $10^3$ times and average the error over independent data sets.
The ``optimal'' interpolation factor is the one that leads to the smallest averaged error.
Repeating this procedure for various sample sizes allows us to describe how the optimal interpolation factor $\alpha$ depends on the ensemble size $n$.

We define the prior covariance matrix $\bm{\Sigma}^{\rm prior}$ using a single-scale, Gaussian kernel model, but we vary the length scale for each set of experiments.
For the single-scale experiments we choose $\ell=1$ (Laplacian kernel) and $\ell=16$ (Gaussian kernel); 
for the multiscale covariance we choose $\ell=4$;
for the nonstationary experiments, we choose $\ell=8$.
Finally, for the pressure-wind experiments we define a prior covariance by 
\begin{equation*}
    \bm{\Sigma}^\text{prior} = \begin{pmatrix}
    \bm{\Sigma}^\ell & \bm{\Sigma}^\ell \\ 
    \bm{\Sigma}^\ell & \bm{\Sigma}^\ell
    \end{pmatrix},
\end{equation*}
where $\bm{\Sigma}^\ell$ is a single-scale Gaussian kernel model with $\ell=5$.
We also performed experiments with different length scales for the various prior covariance matrices and obtained qualitatively similar results, which is not surprising, since the scaling of $\alpha$ in~\cref{eq:ShrinkageAlpha} is largely independent of the choice of the prior covariance matrix.

Results of our experiments are summarized in Figure~\ref{fig:ShrinkResults}.
\begin{figure}[tb]
	\centering
\includegraphics[width=1\textwidth]{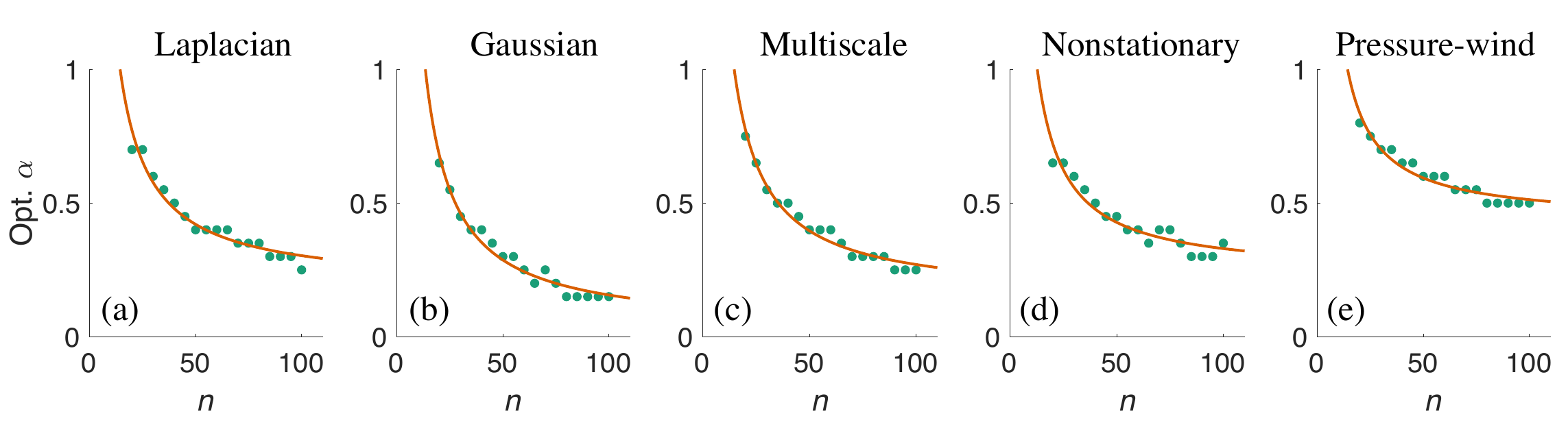}
	\caption{
	Optimal interpolation factor $\alpha$ as a function of ensemble size $n$, with a comparison to the theoretical $\alpha \sim n^{-1}$ scaling.}
	\label{fig:ShrinkResults}
\end{figure}
The figure shows the optimal interpolation factor as a function of the ensemble size $n$, along with a $n^{-1}$ least squares fit to these data.
In all five cases, the fit of a $n^{-1}$ polynomial is quite good, confirming the predictions from the Bayesian theory.
Remarkably, this scaling is independent of the underlying covariance structure.

\subsection{Experiments with Schur product estimators} \label{sec:schur_experiments}
We now consider the Schur product estimator 
\begin{equation*}
    \hat{\bm{\Sigma}}^{\rm Schur} = \hat{\bm{\Sigma}}^{\rm samp} \circ \bm{L}.
\end{equation*}
where the localization matrix is defined by the Laplacian kernel
$\bm{L}_{ij}^{\ell} 
= \exp\bigl(-d_{ij}\slash \ell\bigr)$
or the Gaussian kernel 
$\bm{L}_{ij}^{\ell} 
= \exp\bigl(-(d_{ij}\slash \ell)^2\bigr)$,
involving a single length-scale parameter $\ell$.
For the pressure-wind experiments, we define the overall localization matrix by
\begin{equation*}
    \bm{L} = \begin{pmatrix}
    \bm{\Sigma}^\ell & \bm{\Sigma}^\ell \\ 
    \bm{\Sigma}^\ell & \bm{\Sigma}^\ell
    \end{pmatrix},
\end{equation*}
where $\bm{\Sigma}^\ell$ is a $200\times 200$ single-scale covariance matrix.
We evaluate whether or not the scalings $\ell \sim n$ and $\ell \sim n^{1/2}$ are appropriate, as predicted by the Bayesian theory \cref{eq:SchurScalingG,eq:SchurScalingL}.

We follow the same protocol as in the experiments with the hybrid estimator.
For each ensemble size $n$ and length scale $\ell$,
we perform $10^3$ independent experiments and compute the average of the error defined by 
\begin{equation*}
    \text{Error}(\ell) = \frac{\lVert \hat{\bm{\Sigma}}_\ell^{\rm Schur,\,\ell} - \bm{\Sigma} \rVert_{\rm F}}{\lVert \bm{\Sigma} \rVert_{\rm F}},
\end{equation*}
where $\hat{\bm{\Sigma}}_\ell^{\rm Schur,\,\ell}$ is the Schur product estimate using the length scale $\ell$.
We obtain an optimal length scale by minimizing this error over different length scales $\ell$.
Varying the ensemble size $n$ then allows us to describe how the optimal length scale varies with ensemble size.

\begin{figure}[t]
	\centering
	\includegraphics[width=1\textwidth]{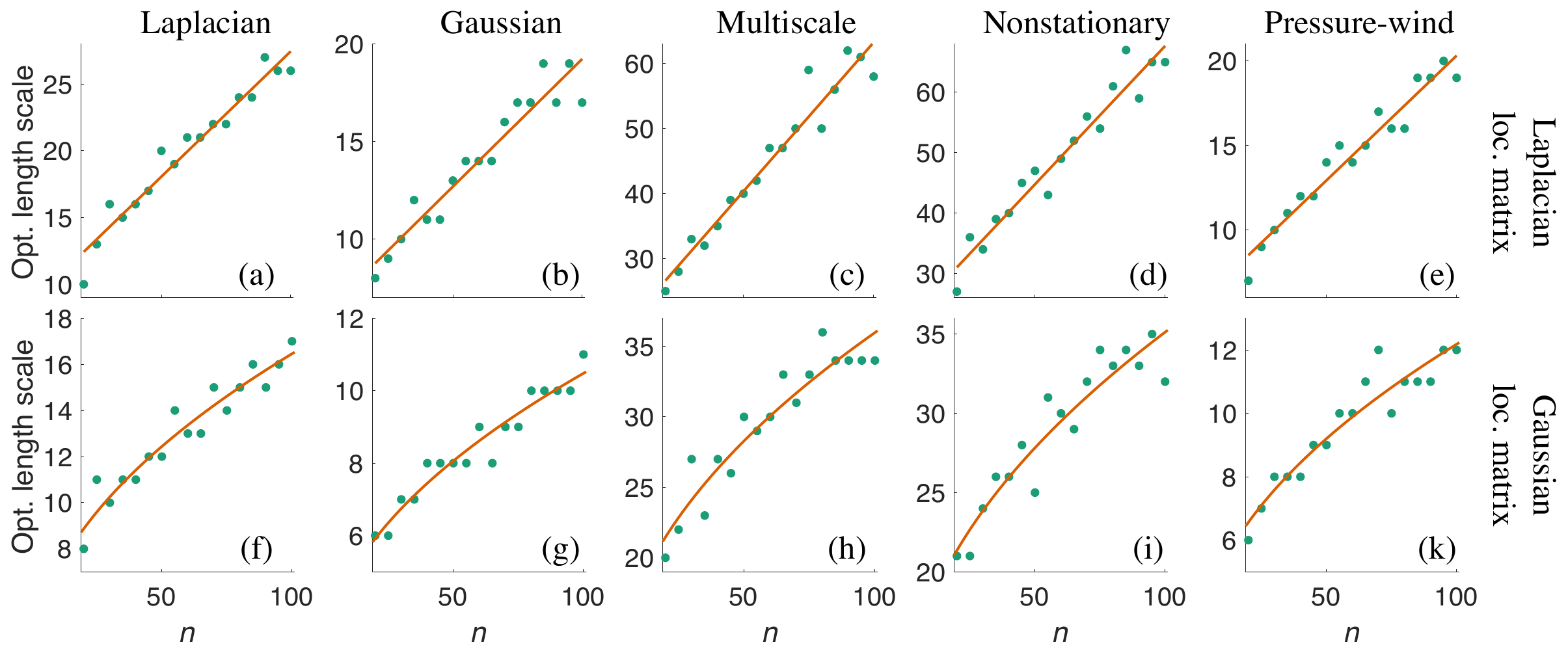}
	\caption{Optimal length scale as a function of ensemble size, with theoretical $\ell \sim n$ (top) and $\ell \sim n^{1/2}$ (bottom) fits.}
	\label{fig:SchurResults}
\end{figure}

Results of our experiments are summarized in \cref{fig:SchurResults}.
The numerical experiments confirm the theoretically derived linear scaling $\ell$ for the Laplacian kernel, as well as the square-root scaling of $\ell$ for the Gaussian kernel.
We see that the theory is useful for finite ensemble size $n$ and finite dimension $d$, and it is robust across five different covariance models.

\section{Properties of the Schur estimator and QC estimator}
\label{sec:theorems}

In this section, we prove mathematical results involving the Schur product estimator and the QC estimator.

\begin{proposition} \label{prop:not_bayesian}
For any positive definite $\bm{L}$,
the Schur product estimator
\begin{equation}
\label{eq:is_map}
    \hat{\bm{\Sigma}}^{\rm Schur} = \hat{\bm{\Sigma}}^{\rm samp} \circ \bm{L}.
\end{equation}
cannot be the MAP estimator for a smooth Bayesian prior distribution.
\end{proposition}
\begin{proof}
We use a proof by contradiction.
If there exists a smooth prior density function whose logarithm is $\ell(\bm{\Sigma})$,
we can write the log posterior density as
\begin{equation}
\label{eq:to_maximize}
	\ell(\bm{\Sigma}) - \frac{n}{2} \log |\bm{\Sigma}| - \frac{n}{2} \textup{tr}(\bm{\Sigma}^{-1} \hat{\bm{\Sigma}}^{\rm samp}).
\end{equation}
The gradient of \cref{eq:to_maximize} is given by
\begin{equation*}
	\nabla \ell(\bm{\Sigma}) - \frac{n}{2} \bm{\Sigma}^{-1}
	+ \frac{n}{2} \bm{\Sigma}^{-1} \hat{\bm{\Sigma}}^{\rm samp} \bm{\Sigma}^{-1}.
\end{equation*}
We next assume that the log posterior density \cref{eq:to_maximize} achieves its maximum value at the matrix $\bm{\Sigma} = \hat{\bm{\Sigma}}^{\rm samp} \circ \bm{L}$.
The Schur product theorem tells us $\bm{\Sigma}$ is positive definite whenever the sample covariance $\hat{\bm{\Sigma}}^{\rm samp}$ is positive definite, whereby $\bm{\Sigma}$ is a local maximum and the gradient equals zero at this point:
\begin{equation}
\label{eq:this_equation}
	\bm{0} = \nabla \ell (\hat{\bm{\Sigma}}^{\rm samp} \circ \bm{L}) 
	- \frac{n}{2} (\hat{\bm{\Sigma}}^{\rm samp} \circ \bm{L})^{-1}
	+ \frac{n}{2} (\hat{\bm{\Sigma}}^{\rm samp} \circ \bm{L})^{-1}
	\hat{\bm{\Sigma}}^{\rm samp}
	(\hat{\bm{\Sigma}}^{\rm samp} \circ \bm{L})^{-1}.
\end{equation}
We will use \cref{eq:this_equation} to obtain a contradiction.
To begin, we differentiate the $(i, j)$ element of \cref{eq:this_equation} with respect to $\hat{\bm{\Sigma}}^{\rm samp}_{kl}$ to yield
\begin{align*}
	0 &= \bm{L}_{kl} 
	\partial^2_{\bm{\Sigma}_{ij} \bm{\Sigma}_{kl}} (\hat{\bm{\Sigma}}^{\rm samp} \circ \bm{L}) 
	+ \frac{n}{2} (\bm{L}_{kl} + 1) (\hat{\bm{\Sigma}}^{\rm samp} \circ \bm{L})^{-1}_{ik}
	(\hat{\bm{\Sigma}}^{\rm samp} \circ \bm{L})^{-1}_{jl} \\
	&- \frac{n}{2} \bm{L}_{kl} 
	(\hat{\bm{\Sigma}}^{\rm samp} \circ \bm{L})^{-1}_{ik}
	\bm{e}_j^\ast (\hat{\bm{\Sigma}}^{\rm samp} \circ \bm{L})^{-1}
	\hat{\bm{\Sigma}}^{\rm samp}
	(\hat{\bm{\Sigma}}^{\rm samp} \circ \bm{L})^{-1} \bm{e}_l \\
	&- \frac{n}{2} \bm{L}_{kl} 
	(\hat{\bm{\Sigma}}^{\rm samp} \circ \bm{L})^{-1}_{jl}
	\bm{e}_i^\ast (\hat{\bm{\Sigma}}^{\rm samp} \circ \bm{L})^{-1}
	\hat{\bm{\Sigma}}^{\rm samp}
	(\hat{\bm{\Sigma}}^{\rm samp} \circ \bm{L})^{-1} \bm{e}_k.
\end{align*}
By multiplying through by $\bm{L}_{ij}$, we obtain an expression of the form
\begin{equation}
\label{eq:symmetric_1}
	\bm{L}_{ij} \bm{L}_{kl} \partial^2_{\bm{\Sigma}_{ij} \bm{\Sigma}_{kl}} (\hat{\bm{\Sigma}}^{\rm samp} \circ \bm{L}) = - \frac{n}{2} \bm{L}_{ij} (\hat{\bm{\Sigma}}^{\rm samp} \circ \bm{L})^{-1}_{ik}
	(\hat{\bm{\Sigma}}^{\rm samp} \circ \bm{L})^{-1}_{jl} + f_{\rm sym}((i, j), (k, l)),
\end{equation}
where $f_{\rm sym}((i, j), (k, l))$ is a symmetric function of $(i, j)$ and $(k, l)$. Switching the roles of $(i, j)$ and $(k, l)$ and using the symmetry of the second derivatives, we have
\begin{equation}
\label{eq:symmetric_2}
	\bm{L}_{ij} \bm{L}_{kl} \partial^2_{\bm{\Sigma}_{ij} \bm{\Sigma}_{kl}} (\hat{\bm{\Sigma}}^{\rm samp} \circ \bm{L}) = - \frac{n}{2} \bm{L}_{kl} (\hat{\bm{\Sigma}}^{\rm samp} \circ \bm{L})^{-1}_{jl}
	(\hat{\bm{\Sigma}}^{\rm samp} \circ \bm{L})^{-1}_{jl} + f_{\rm sym}((i, j), (k, l)).
\end{equation}
Subtracting \cref{eq:symmetric_1} from \cref{eq:symmetric_2} leads to the expression
\begin{equation*}
	0 = (\bm{L}_{ij} - \bm{L}_{kl}) (\hat{\bm{\Sigma}}^{\rm samp} \circ \bm{L})^{-1}_{ik}
	(\hat{\bm{\Sigma}}^{\rm samp} \circ \bm{L})^{-1}_{jl},
\end{equation*}
valid for all positive definite matrices $\hat{\bm{\Sigma}}^{\rm samp}$ and all indices $1 \leq i, j, k, l \leq p$.
We note that $(\hat{\bm{\Sigma}}^{\rm samp} \circ \bm{L})^{-1}_{ii}$ must be positive (it is the diagonal entry of a positive definite matrix); hence,
\begin{equation}
\label{eq:differentiate_again}
	0 = (\bm{L}_{ij} - \bm{L}_{il}) (\hat{\bm{\Sigma}}^{\rm samp} \circ \bm{L})^{-1}_{jl},
\end{equation}
Differentiating \cref{eq:differentiate_again} with respect to $\hat{\bm{\Sigma}}^{\rm samp}_{jl}$, we find
\begin{equation*}
	0 = -(\bm{L}_{ij} - \bm{L}_{il})
	(\bm{\Sigma} \circ \bm{L})^{-1}_{jj} 
	(\bm{\Sigma} \circ \bm{L})^{-1}_{ll} \bm{L}_{jl},
\end{equation*}
which implies that the $j$ and $l$ columns of $\bm{L}$ are identical when $\bm{L}_{jl} \neq 0$.
Since $\bm{L}$ is positive definite, we cannot have two identical columns and we conclude that $\bm{L}$ must be a diagonal matrix.
Last, $\bm{L}$ cannot be diagonal. To show this,
we write $\hat{\bm{\Sigma}}^{\rm samp} = \bm{D} + \epsilon \bm{\Delta}$ where $\bm{D}$ is diagonal and $\bm{\Delta}$ is a nontrivial off-diagonal matrix.
Then, the gradient equality \cref{eq:this_equation} yields
\begin{equation*}
	\bm{0} = \nabla \ell (\bm{D} \circ \bm{L}) 
	- \frac{n}{2} (\bm{D} \circ \bm{L})^{-1}
	+ \frac{n}{2} (\bm{D} \circ \bm{L})^{-1}
	(\bm{D} + \epsilon \bm{\Delta})
	(\bm{D} \circ \bm{L})^{-1}.
\end{equation*}
Taking the derivative with respect to $\epsilon$, we find that $\bm{\Delta} = \bm{0}$, which cannot hold because $\bm{\Delta}$ is nontrivial by assumption.
We have arrived at a contradiction and conclude that the Schur product estimator \cref{eq:is_map} cannot be the MAP estimator for a smooth Bayesian prior.
\end{proof}

\begin{proposition} \label{prop:max_entropy}
If the QC density 
\begin{equation}
\label{eq:to_confirm}
    p(\bm{\Sigma}) \propto
    \exp\Bigl(- \frac{1}{4} \textup{tr}\bigl(\bm{\Sigma}^{-1} \bigr(\bm{\Theta} \circ \bm{\Sigma}^{-1}\bigr)\bigr) \Bigr)
\end{equation}
is a well-defined density (i.e., it integrates to one),
then it solves the entropy maximization problem
\begin{equation*}
    \max_p 
    \biggl\{
    H[p]
    - \frac{1}{4} \sum_{i,j = 1}^d \bm{\Theta}_{ij}\, \int p(\bm{\Sigma}) |\bm{\Sigma}^{-1}_{ij}|^2 d\bm{\Sigma} \biggr\}.
\end{equation*}
\end{proposition}
\begin{proof}
We introduce a Lagrange multiplier $\eta$ to enforce the constraint that $\int p(\bm{\Sigma}) \mathop{d \bm{\Sigma}} = 1$ and then choose $p$ to maximize
\begin{equation*}
	L[p] = H[p]
	- \frac{1}{4} \sum_{i,j = 1}^d \bm{\Theta}_{ij}\, \int p(\bm{\Sigma}) |\bm{\Sigma}^{-1}_{ij}|^2 d\bm{\Sigma}
	+ \eta \Bigl(\int p(\bm{\Sigma}) \mathop{d \bm{\Sigma}} - 1\Bigr).
\end{equation*}
The first variation of $L[p]$ is
\begin{equation*}
	\delta L [p] = -\log p(\bm{\Sigma}) - 1 - 
	\frac{1}{4} \textup{tr}\bigl(\bm{\Sigma}^{-1} \bigr(\bm{\Theta} \circ \bm{\Sigma}^{-1}\bigr)\bigr)  + \eta.
\end{equation*}
Setting the first variation equal to zero yields
\begin{equation*}
	p(\bm{\Sigma}) = \exp\Bigl(- \frac{1}{4} \textup{tr}\bigl(\bm{\Sigma}^{-1} \bigr(\bm{\Theta} \circ \bm{\Sigma}^{-1}\bigr)\bigr) - 1 + \eta \Bigr),
\end{equation*}
which is a scalar multiple of the QC density \cref{eq:to_confirm}.
By selecting an appropriate Lagrange multiplier $\eta$, we ensure that $p$ is a probability density.
Additionally, $p \mapsto L[p]$ is a concave functional on probability densities, because the entropy $H[p]$ is concave on probability densities and the other terms are linear in $p$.
Since \cref{eq:to_confirm} is a stationary point of a concave Lagrangian, we conclude that \cref{eq:to_confirm} maximizes $L[p]$ and solves the entropy maximization problem.
\end{proof}

\begin{proposition} \label{prop:maximum}
Assume $\textup{diag}(\hat{\bm{\Sigma}}^{\textup{samp}}) > 0$, and consider a penalization of the form 
$\bm{\Theta} = s \bm{\Theta}^{\rm ref}$, 
where $\bm{\Theta}_{ij}^{\rm ref} > 0$ if and only if $i \neq j$.
Then for large enough $s > 0$, the log likelihood
\begin{equation*}
    \ell(\bm{\Sigma}) = 
    \frac{n}{2} \log \bigl| \bm{\Sigma}^{-1} \bigr| - \frac{n}{2} \textup{tr}\bigl(\hat{\bm{\Sigma}}^{\rm samp} \bm{\Sigma}^{-1}\bigr) - \frac{1}{4} \textup{tr}\bigl(\bm{\Sigma}^{-1} \bigl( \bm{\Theta} \circ \bm{\Sigma}^{-1} \bigr) \bigr)
\end{equation*}
has a unique positive definite global maximizer.
\end{proposition}
\begin{proof}
We decompose the sample covariance into diagonal and off-diagonal components 
\begin{equation*}
    \hat{\bm{\Sigma}}^{\rm samp} = \hat{\bm{\Sigma}}^{\rm diag} + \hat{\bm{\Sigma}}^{\rm off}
\end{equation*}
and similarly decompose the log likelihood as
\begin{equation*}
    \ell(\bm{\Sigma}) = 
    \Biggl[\frac{n}{2} \log \bigl| \bm{\Sigma}^{-1} \bigr| - \frac{n}{2} \textup{tr}\bigl(\hat{\bm{\Sigma}}^{\rm diag} \bm{\Sigma}^{-1}\bigr)\Biggr]
    + \Biggl[- \frac{n}{2} \textup{tr}\bigl(\hat{\bm{\Sigma}}^{\rm off} \bm{\Sigma}^{-1}\bigr)
    - \frac{1}{4} \textup{tr}\bigl(\bm{\Sigma}^{-1} \bigl( \bm{\Theta} \circ \bm{\Sigma}^{-1} \bigr) \bigr)\Biggr].
\end{equation*}
Next, we optimize the two components of the log likelihood separately.
Since the determinant of a positive definite matrix is bounded from above by the product of its diagonal entries, we calculate
\begin{equation*}
    \frac{n}{2} \log \bigl| \bm{\Sigma}^{-1} \bigr| 
	- \frac{n}{2} \textup{tr}\bigl(\hat{\bm{\Sigma}}^{\rm diag} \bm{\Sigma}^{-1}\bigr)
	\leq \frac{n}{2} \sum_{i=1}^d \bigl[\log \bm{\Sigma}^{-1}_{ii} - \hat{\bm{\Sigma}}^{\rm samp}_{ii} \bm{\Sigma}^{-1}_{ii}\bigr] 
	\leq \frac{n}{2} \sum_{i=1}^d \bigl[-\log \hat{\bm{\Sigma}}^{\rm samp}_{ii} - 1\bigr].
\end{equation*}
The left-hand side does not come within $\frac{n}{2}$ units of the upper bound unless
\begin{equation*}
    \frac{1}{2} \leq \hat{\bm{\Sigma}}^{\rm samp}_{ii} \bm{\Sigma}^{-1}_{ii} \leq 2
\end{equation*}
for each $1 \leq i \leq d$.
Next observe that 
\begin{align*}
    - \frac{n}{2} \textup{tr}\bigl(\hat{\bm{\Sigma}}^{\rm off} \bm{\Sigma}^{-1}\bigr)
    - \frac{1}{4} \textup{tr}\bigl(\bm{\Sigma}^{-1} \bigl( \bm{\Theta} \circ \bm{\Sigma}^{-1} \bigr) \bigr)
    &= \frac{1}{4} \sum_{i \neq j} 
    \biggl[\frac{n^2 |\hat{\bm{\Sigma}}_{ij}^{\rm samp}|^2}{\bm{\Theta}_{ij}} 
    - \bm{\Theta}_{ij} \biggl|\bm{\Sigma}_{ij}^{-1} + \frac{n \hat{\bm{\Sigma}}_{ij}^{\rm samp}}{\bm{\Theta}_{ij}} \biggr|^2 \biggr] \\
    &\leq \frac{n^2}{4} \sum_{i \neq j} 
    \frac{|\hat{\bm{\Sigma}}_{ij}^{\rm samp}|^2}{\bm{\Theta}_{ij}}.
\end{align*}
In this case, the left-hand side does not come within $\frac{n}{2}$ units of the upper bound, unless
\begin{equation*}
    \biggl|\bm{\Sigma}_{ij}^{-1} + \frac{n \hat{\bm{\Sigma}}_{ij}^{\rm samp}}{\bm{\Theta}_{ij}} \biggr|^2 \leq \frac{2n}{\bm{\Theta}_{ij}}
\end{equation*}
for each $i \neq j$.
Tying together the argument, a large log likelihood
\begin{equation*}
    \ell(\bm{\Sigma}) \geq \frac{n}{2} \sum_{i=1}^d \bigl[-\log \hat{\bm{\Sigma}}^{\rm samp}_{ii} - 1\bigr] + \frac{n^2}{4} \sum_{i \neq j} 
    \frac{|\hat{\bm{\Sigma}}_{ij}^{\rm samp}|^2}{\bm{\Theta}_{ij}}
    - \frac{n}{2}
\end{equation*}
implies the following conditions
\begin{equation*}
    \begin{cases}
    \frac{1}{2} \leq \hat{\bm{\Sigma}}^{\rm samp}_{ii} \bm{\Sigma}^{-1}_{ii} \leq 2, & 1 \leq i \leq d \\[10pt]
    \biggl|\bm{\Sigma}_{ij}^{-1} + \frac{n \hat{\bm{\Sigma}}_{ij}^{\rm samp}}{\bm{\Theta}_{ij}} \biggr|^2 \leq \frac{2n}{\bm{\Theta}_{ij}}, & i \neq j.
    \end{cases}
\end{equation*}
For a large enough penalization strength $s$, these conditions confine $\bm{\Sigma}$ to a compact set of positive definite matrices.
Additionally, this set contains an isolated maximizer of $\ell$, because one point inside the set, $\bm{\Sigma} = \hat{\bm{\Sigma}}^{\rm diag}$, achieves a higher log likelihood
\begin{equation*}
    \ell(\hat{\bm{\Sigma}}^{\rm diag}) = \frac{n}{2} \sum_{i=1}^d \bigl[-\log \hat{\bm{\Sigma}}^{\rm samp}_{ii} - 1\bigr]
\end{equation*}
than all the points outside the set.
Last, note that
\begin{equation*}
	\bm{\Sigma}^{-1} \mapsto \frac{n}{2} \log \bigl| \bm{\Sigma}^{-1} \bigr| 
	- \frac{n}{2} \textup{tr}\bigl(\hat{\bm{\Sigma}}^{\rm samp} \bm{\Sigma}^{-1}\bigr) 
	- \frac{1}{4} \textup{tr}\bigl(\bm{\Sigma}^{-1} \bigl( \bm{\Theta} \circ \bm{\Sigma}^{-1} \bigr) \bigr)
\end{equation*}
is concave, whereby any isolated maximum is the unique global maximum.
\end{proof}

\section{Conclusions} \label{sec:conclusions}
We have studied the problem of estimating a high-dimensional covariance matrix from a small number of samples.
This problem is difficult but also ubiquitous in the setting of numerical weather prediction (NWP)
when merging high-dimensional models with real-world observations.
We have developed a new mathematical theory to help justify the covariance estimators which are used in NWP practice, but which have little to no mathematical justification.

NWP practitioners boost the accuracy of covariance estimates by enforcing the assumption that correlations decay with distance.
We have argued that this practical approach follows the Bayesian paradigm of estimating an unknown (covariance matrix) from data (the samples we have) and prior information (spatial decay of correlations).
We have then investigated prior distributions that lead to practically useful covariance estimators.
Put differently, we have interpreted common estimation procedures within a Bayesian framework and identified how particular choices of prior distributions influence posterior estimates of covariance matrices.

The situation is clear in the case of hybrid estimators. 
These estimators can be understood as Bayesian estimators with an inverse Wishart prior distribution. 
The Bayesian theory shows that hybrid estimators converge at an asymptotically optimal rate and also reveals how to adjust the hybrid estimator when the ensemble size changes.

The Bayesian interpretation of Schur product estimators is more complicated.
We have shown that Schur product estimators are \emph{not} Bayesian.
Nonetheless, we have proposed a new ``quadratically constrained'' distribution that leads to a Schur product estimator in the limit of increasing localization strength.
Building on this interpretation, the Bayesian theory suggests how best to adjust the length scale in a Schur product estimator as the ensemble size increases.
This scaling is new and relevant for practical NWP because it can reduce the amount of tuning required for efficient covariance estimation.

Perhaps more importantly, our analysis suggests that covariance estimation via Schur products may not be the most economical approach for covariance estimation.
For one, it is \emph{not} Bayesian and therefore it has not been rigorously justified.
Additionally, our theory, in line with previous work \cite{MTM19}, suggests that we can estimate covariance matrices more efficiently by working with the inverse of the covariance matrix, i.e., by penalizing \emph{conditional} correlations, rather than correlations.

\section*{Acknowledgments}
We thank Dr. Daniel Hodyss of the Naval Research Laboratory for discussions of hybrid estimators.

\bibliographystyle{siamplain}
\bibliography{references}

\end{document}

%% file: ex_shared.tex
\usepackage{lipsum}
\usepackage{amsfonts}
\usepackage{graphicx}
\usepackage{epstopdf}
\usepackage{algorithmic}
\usepackage{bm}
\usepackage[sort,compress]{cite}

\ifpdf
  \DeclareGraphicsExtensions{.eps,.pdf,.png,.jpg}
\else
  \DeclareGraphicsExtensions{.eps}
\fi

\usepackage{enumitem}
\setlist[enumerate]{leftmargin=.5in}
\setlist[itemize]{leftmargin=.5in}

\newsiamremark{remark}{Remark}
\newsiamremark{question}{Question}
\newsiamremark{lesson}{Lesson}
\newsiamremark{hypothesis}{Hypothesis}
\crefname{hypothesis}{Hypothesis}{Hypotheses}
\newsiamthm{claim}{Claim}

\headers{Localized covariance estimation: A Bayesian justification}{R.J. Webber and M. Morzfeld}

\title{Localized covariance estimation: A Bayesian justification\thanks{Submitted to the editors DATE.
\funding{RJW is supported by the Office of Naval Research through BRC award N00014-18-1-2363 and the National Science Foundation through FRG award 1952777, under the aegis of Joel A. Tropp.
MM is supported by the US Office of Naval Research (ONR) grant N00014-21-1-2309.}}}

\author{Robert J. Webber\thanks{Computing \& Mathematical Sciences, California Institute of Technology, Pasadena, CA
  (\email{rwebber@caltech.edu}).}
\and Matthias Morzfeld\thanks{Institute of Geophysics and Planetary Physics, Sripps Institution of Oceanography, University of California, San Diego, San Diego, CA
  (\email{matti@ucsd.edu}).}}

\usepackage{amsopn}

\makeatletter
\newcommand*{\addFileDependency}[1]{
  \typeout{(#1)}
  \@addtofilelist{#1}
  \IfFileExists{#1}{}{\typeout{No file #1.}}
}
\makeatother

\newcommand*{\myexternaldocument}[1]{%
    \externaldocument{#1}%
    \addFileDependency{#1.tex}%
    \addFileDependency{#1.aux}%
}

\DeclareMathOperator*{\argmax}{arg\,max}


%% file: ex_article.bbl
\begin{thebibliography}{10}

\bibitem{AL13}
{\sc J.~Anderson and L.~Lei}, {\em Empirical localization of observation impact
  in ensemble {K}alman filters}, Monthly Weather Review, 141 (2013),
  pp.~4140--4153, \url{https://doi.org/10.1175/MWR-D-12-00330.1}.

\bibitem{BTB14}
{\sc P.~Bauer, A.~Thorpe, and G.~Brunet}, {\em The quiet revolution of
  numerical weather prediction}, Nature,  (2015), pp.~47--55,
  \url{https://doi.org/10.1038/nature14956}.

\bibitem{BL08a}
{\sc P.~J. Bickel and E.~Levina}, {\em {Covariance regularization by
  thresholding}}, Annals of Statistics, 36 (2008), pp.~2577--2604,
  \url{https://doi.org/10.1214/08-AOS600}.

\bibitem{BL08}
{\sc P.~J. Bickel and E.~Levina}, {\em {Regularized estimation of large
  covariance matrices}}, Annals of Statistics, 36 (2008), pp.~199--227,
  \url{https://doi.org/10.1214/009053607000000758}.

\bibitem{BS13}
{\sc C.~H. Bishop and E.~A. Satterfield}, {\em Hidden error variance theory.
  {P}art {I}: {E}xposition and analytic model}, Monthly Weather Review, 141
  (2013), pp.~1454--1468, \url{https://doi.org/10.1175/MWR-D-12-00118.1}.

\bibitem{CRZ16}
{\sc T.~T. Cai, Z.~Ren, and H.~H. Zhou}, {\em {Estimating structured
  high-dimensional covariance and precision matrices: {O}ptimal rates and
  adaptive estimation}}, Electronic Journal of Statistics, 10 (2016),
  pp.~1--59, \url{https://doi.org/10.1214/15-EJS1081}.

\bibitem{dawid1973marginalization}
{\sc A.~P. Dawid, M.~Stone, and J.~V. Zidek}, {\em Marginalization paradoxes in
  {B}ayesian and structural inference}, Journal of the Royal Statistical
  Society. Series B (Methodological), 35 (1973), pp.~189--233,
  \url{http://www.jstor.org/stable/2984907}.

\bibitem{diaconis1979conjugate}
{\sc P.~Diaconis and D.~Ylvisaker}, {\em {Conjugate Priors for Exponential
  Families}}, Annals of Statistics, 7 (1979), pp.~269--281,
  \url{https://doi.org/10.1214/aos/1176344611}.

\bibitem{E09Book}
{\sc G.~Evensen}, {\em Data Assimilation: The Ensemble Kalman Filter},
  Springer, second~ed., 2009, \url{https://doi.org/10.1007/978-3-642-03711-5}.

\bibitem{F15}
{\sc J.~Flowerdew}, {\em Towards a theory of optimal localisation}, Tellus A:
  Dynamic Meteorology and Oceanography, 67 (2015), p.~25257,
  \url{https://doi.org/10.3402/tellusa.v67.25257}.

\bibitem{FHetAl10}
{\sc A.~Fournier, G.~Hulot, D.~Jault, W.~Kuang, A.~Tangborn, N.~Gillet,
  E.~Canet, J.~Aubert, and F.~Lhuillier}, {\em An introduction to data
  assimilation and predictability in geomagnetism}, Space Science Review, 155
  (2010), pp.~247--291, \url{https://doi.org/10.1007/s11214-010-9669-4}.

\bibitem{FB07}
{\sc R.~Furrer and T.~Bengtsson}, {\em Estimation of high-dimensional prior and
  posterior covariance matrices in {K}alman filter variants}, Journal of
  Multivariate Analysis, 98 (2007), pp.~227--255,
  \url{https://doi.org/10.1016/j.jmva.2006.08.003}.

\bibitem{GC99}
{\sc G.~Gaspari and S.~E. Cohn}, {\em Construction of correlation functions in
  two and three dimensions}, Quarterly Journal of the Royal Meteorological
  Society, 125 (1999), pp.~723--757,
  \url{https://doi.org/10.1002/qj.49712555417}.

\bibitem{GMTK21}
{\sc K.~Gwirtz, M.~Morzfeld, W.~Kuang, and A.~Tangborn}, {\em {A testbed for
  geomagnetic data assimilation}}, Geophysical Journal International, 227
  (2021), pp.~2180--2203, \url{https://doi.org/10.1093/gji/ggab327}.

\bibitem{HWAS09}
{\sc T.~M. Hamill, J.~S. Whitaker, J.~L. Anderson, and C.~Snyder}, {\em
  Comments on “sigma-point kalman filter data assimilation methods for
  strongly nonlinear systems”}, Journal of the Atmospheric Sciences, 66
  (2009), pp.~3498--3500, \url{https://doi.org/10.1175/2009JAS3245.1}.

\bibitem{HM98}
{\sc P.~L. Houtekamer and H.~L. Mitchell}, {\em Data assimilation using an
  ensemble {K}alman filter technique}, Monthly Weather Review, 126 (1998),
  pp.~796--811,
  \url{https://doi.org/10.1175/1520-0493(1998)126<0796:DAUAEK>2.0.CO;2}.

\bibitem{H01}
{\sc P.~L. Houtekamer and H.~L. Mitchell}, {\em A sequential ensemble {K}alman
  filter for atmospheric data assimilation}, Monthly Weather Review, 129
  (2001), pp.~123--137,
  \url{https://doi.org/10.1175/1520-0493(2001)129<0123:ASEKFF>2.0.CO;2}.

\bibitem{jaynes1957information}
{\sc E.~T. Jaynes}, {\em Information theory and statistical mechanics},
  Physical Review, 106 (1957), pp.~620--630,
  \url{https://doi.org/10.1103/PhysRev.106.620}.

\bibitem{L03}
{\sc A.~C. Lorenc}, {\em The potential of the ensemble {K}alman filter for
  {NWP}—a comparison with {4D-Var}}, Quarterly Journal of the Royal
  Meteorological Society, 129 (2003), pp.~3183--3203,
  \url{https://doi.org/https://doi.org/10.1256/qj.02.132}.

\bibitem{MH22}
{\sc M.~Morzfeld and D.~Hodyss}, {\em A theory for why even simple covariance
  localization is so useful in ensemble data assimilation}, Monthly Weather
  Review,  (2022), \url{https://doi.org/10.1175/MWR-D-22-0255.1}.

\bibitem{MTM19}
{\sc M.~Morzfeld, X.~Tong, and Y.~Marzouk}, {\em Localization for {MCMC}:
  {S}ampling high-dimensional posterior distributions with local structure},
  Journal of Computational Physics, 380 (2019), pp.~1--28,
  \url{https://doi.org/https://doi.org/10.1016/j.jcp.2018.12.008}.

\bibitem{MA15}
{\sc B.~Ménétrier and T.~Auligné}, {\em Optimized localization and
  hybridization to filter ensemble-based covariances}, Monthly Weather Review,
  143 (2015), pp.~3931--3947, \url{https://doi.org/10.1175/MWR-D-15-0057.1}.

\bibitem{Menetrier15}
{\sc B.~Ménétrier, T.~Montmerle, Y.~Michel, and L.~Berre}, {\em Linear
  filtering of sample covariances for ensemble-based data assimilation. {P}art
  {I}: {O}ptimality criteria and application to variance filtering and
  covariance localization}, Monthly Weather Review, 143 (2015), pp.~1622--1643,
  \url{https://doi.org/10.1175/MWR-D-14-00157.1}.

\bibitem{NRS15}
{\sc E.~Nino-Ruiz and A.~Sandu}, {\em Ensemble {K}alman filter implementations
  based on shrinkage covariance matrix estimation}, Ocean Dynamics, 65 (2015),
  p.~1423–1439, \url{https://doi.org/10.1007/s10236-015-0888-9}.

\bibitem{OEtAl04}
{\sc E.~Ott, B.~Hunt, I.~Szunyogh, A.~Zimin, E.~Kostelich, M.~Corazza,
  E.~Kalnay, D.~Patil, and J.~Yorke}, {\em A local ensemble {K}alman filter for
  atmospheric data assimilation}, Tellus A, 56 (2004), pp.~415--428,
  \url{https://doi.org/10.3402/tellusa.v56i5.14462}.

\bibitem{PS03}
{\sc C.~Paciorek and M.~Schervish}, {\em Nonstationary covariance functions for
  {G}aussian process regression}, in Advances in Neural Information Processing
  Systems, S.~Thrun, L.~Saul, and B.~Sch\"{o}lkopf, eds., vol.~16, MIT Press,
  2003.

\bibitem{PSS21}
{\sc A.~A. Popov, A.~N. Subrahmanya, and A.~Sandu}, {\em A stochastic
  covariance shrinkage approach to particle rejuvenation in the ensemble
  transform particle filter}, Nonlinear Processes in Geophysics, 29 (2022),
  pp.~241--253, \url{https://doi.org/10.5194/npg-29-241-2022}.

\bibitem{PZ15}
{\sc J.~Poterjoy and F.~Zhang}, {\em Systematic comparison of four-dimensional
  data assimilation methods with and without the tangent linear model using
  hybrid background error covariance: {E4DVar} versus {4DEnVar}}, Monthly
  Weather Review, 143 (2015), pp.~1601--1621,
  \url{https://doi.org/10.1175/MWR-D-14-00224.1}.

\bibitem{press2005applied}
{\sc S.~J. Press}, {\em Applied Multivariate Analysis: Using Bayesian and
  Frequentist Methods of Inference}, Dover, second~ed., 2005.

\bibitem{SHKB18}
{\sc E.~A. Satterfield, D.~Hodyss, D.~D. Kuhl, and C.~H. Bishop}, {\em
  Observation-informed generalized hybrid error covariance models}, Monthly
  Weather Review, 146 (2018), pp.~3605--3622,
  \url{https://doi.org/10.1175/MWR-D-18-0016.1}.

\bibitem{stein2002screening}
{\sc M.~L. Stein}, {\em The screening effect in kriging}, Annals of Statistics,
  30 (2002), pp.~298--323, \url{http://www.jstor.org/stable/2700012}.

\bibitem{van2000asymptotic}
{\sc A.~W. v.~d. Vaart}, {\em Asymptotic Statistics}, Cambridge Series in
  Statistical and Probabilistic Mathematics, Cambridge University Press, 1998,
  \url{https://doi.org/10.1017/CBO9780511802256}.

\end{thebibliography}
